\RequirePackage{fix-cm}
\documentclass[smallextended]{svjour3}       % onecolumn (second format)
\smartqed  % flush right qed marks, e.g. at end of proof
\usepackage{latexsym}
\usepackage{mathrsfs}
\usepackage{multirow}

\usepackage{amsfonts,amssymb,amsmath,amsthm,bm}
\usepackage{color}
\usepackage{graphics}
\usepackage{makecell}
\usepackage{float}
\usepackage{autobreak} 
\usepackage{booktabs}

\usepackage[figuresright]{rotating}
\usepackage{array}
\allowdisplaybreaks
\def\BF{\mathbb{F}}

\def\bv{\mathbf{v}}
\def\bu{\mathbf{u}}
\def\bH{\mathbf{H}}
\def\bO{\mathbf{O}}
\def\bI{\mathbf{I}}
\def\bA{\mathbf{A}}
\def\bP{\mathbf{P}}
\def\bQ{\mathbf{Q}}
\def\bM{\mathbf{M}}

\def\bB{\mathbf{B}}
\def\bV{\mathbf{V}}

\def\bU{\mathbf{U}}

\def\cF{\mathcal{F}}
\def\cD{\mathcal{D}}
\def\cS{\mathcal{S}}
\def\cC{\mathcal{C}}

\def\cH{\mathcal{H}}
\def\cP{\mathcal{P}}
\def\cQ{\mathcal{Q}}
\def\cX{\mathcal{X}}

\def\cX{\mathcal{X}}

\def\cF{\mathcal{F}}
\def\cG{\mathcal{G}}

\def\tcC{\Tilde{\mathcal{C}}}

\def\tcS{\tilde{\mathcal{S}}}
\def\tcF{\Tilde{\mathcal{F}}}
\def\tbv{\Tilde{\mathbf{v}}}
\def\tbu{\Tilde{\mathbf{u}}} 

\def\obv{\bar{\mathbf{v}}}
\def\obu{\bar{\mathbf{u}}}

\def\sL{\mathscr{L}}
\def\rank{\mathrm{rank}}
\def\rs{\mathrm{rs}}
\def\wt{\mathrm{wt}}
\def\RREF{\mathrm{RREF}}
\def\EF{\mathrm{EF}}

\newtheorem{construction}{Construction}
\journalname{}
\begin{document}
	
	\title{Generalized bilateral multilevel construction for constant dimension codes from parallel mixed dimension construction \thanks{This research is supported by the National Key Research and Development Program of China (Grant No. 2022YFA1005000), the National Natural Science Foundation of China (Grant Nos. 12141108, 62371259, and 12411540221), the Fundamental Research Funds for the Central Universities of China (Nankai University), and the Nankai Zhide Foundation.}
	}
	
	
	\author{Han Li  \and
		Fang-Wei Fu%etc.
	}
	

	\institute{Han Li\at
		Chern Institute of Mathematics and LPMC, Nankai University, Tianjin, 300071, China\\
		\email{hli@mail.nankai.edu.cn}             \\
	%	\and 
	%	Jiaxin Wang\at
	%	Chern Institute of Mathematics and LPMC, Nankai University, Tianjin, 300071, China\\
		%\email{wjiaxin@mail.nankai.edu.cn}
		\and
		Fang-Wei Fu\at
		Chern Institute of Mathematics and LPMC, Nankai University, Tianjin, 300071, China\\
		\email{fwfu@nankai.edu.cn}}
	\date{Received: date / Accepted: date}

	\maketitle
	\begin{abstract}
 Constant dimension codes (CDCs), as special subspace codes, have received extensive attention due to their
applications in random network coding. The basic problem of CDCs is to determine the maximal possible size $A_q(n,d,\{k\})$  for given parameters $q, n, d$, and $k$. This paper introduces criteria for choosing appropriate bilateral identifying vectors compatible with the parallel mixed dimension construction (Des. Codes Cryptogr. 93(1):227--241, 2025). We then utilize the generalized bilateral multilevel construction (Des. Codes Cryptogr. 93(1):197--225, 2025) to improve the parallel mixed dimension construction efficiently. Many new CDCs that are better than the previously best-known codes are constructed. 
	\end{abstract}

\section{Introduction}
Subspace codes have important implications for error correction in random linear network coding by K\"{o}etter and  Kschischang \cite{Network}. For any prime power $q \geq 2$, let $\BF_q$ be the finite field with $q$ elements, and $\BF_q^n$ be the $n$-dimensional vector space over $\BF_q$.  The set of all $\BF_q$-subspaces of $\BF_q^n$ is denoted by $\cP_q(n)$. Given a nonnegative integer $ k \leq n$, the set of all $k$-dimensional subspaces in $\cP_q(n)$ is known as the Grassmannian $\cG_q(n,k)$. The  cardinality of $\cG_q(n,k)$ is given by the $q$-ary Gaussian binomial coefficient \[\begin{bmatrix}n\\ k     \end{bmatrix}_{q}= \prod \limits_{i = 0}^{k - 1}\frac{q^n - q^i}{q^k-q^i}.\]
It is obvious that $\cP_q(n) = \cup_{k=0}^n \cG_q(n,k)$, which forms a metric space \cite{Network} equipped with the subspace metric $d_{S}(U,V) \triangleq \dim(U + V) - \dim(U \cap V)$,  where $\dim(\cdot)$ denotes the dimension of a vector space over $\mathbb{F}_{q}$.
A subspace code $\cC$ is a nonempty subset of $\cP_q(n)$. Elements in $\cC$ are called codewords.  A subspace code $\cC$ 
 is called an $(n,d)_q$ subspace code if the subspace distance between each pair of distinct codewords in $\cC$ is at least $d$. The dimension distribution of $\cC$ is denoted by $\eta_0(\cC)$, $\eta_1(\cC)$, $\dots$, $\eta_n(\cC)$, where $\eta_i(\cC)$ is defined as 
\[
|\{X \in \cC: \dim(X) = i\}|.
\] Here, $|\cdot|$ denotes the cardinality of the set. Since a code typically has many values satisfying $\eta_i(\cC)=0$, these values are often omitted from the list. Consider a set $T\subseteq \{0,1, \dots, n\}$, then we say that $\cC$ is an $(n,d,T)_q$  subspace code if $\eta_i(\cC) > 0$ for $i \in T$ and $\eta_i(\cC)=0$ for $i \notin T$. If $|T| >1$, this code is called a mixed dimension code (MDC). If $|T| =1$, this code is called a constant dimension code (CDC). Additionally, an $(n, d, T)_q$ subspace code with $M$ codewords can be denoted as an $(n, M, d, T)_q$ subspace code. Let $A_q(n,d,\{k\})$ denote the maximum possible size of an $(n,d,\{k\})_q$ CDC. An $(n,d,\{k\})_q$ CDC with $A_q(n,d,\{k\})$ codewords is said to be optimal. One of the main problems in the study of CDCs is to determine the exact value of $A_q(n,d,\{k\})$. The homepage http://subspacecodes.uni-bayreuth.de/  lists the currently best-known lower and upper bounds on $A_q(n,d,\{k\})$ for $q \leq 9$, $n \leq 19$, all $d$, and all $k$ \cite{Table}.

Rank-metric codes, particularly maximum rank distance (MRD) codes, are usually used to construct CDCs. The lifting construction is an effective method for constructing CDCs, and lifted MRD codes can result in asymptotically optimal CDCs \cite{Lifting}. By introducing identifying vectors and a new family of Ferrers diagram rank-metric codes, Etzion and Silberstein \cite{Multilevel} proposed the multilevel construction to generalize the lifting construction. However, how
to choose identifying vectors effectively remains an open problem. Using the idea of pending dots and graph matching, large subspace codes in $\cG_q(n,k)$ with minimum subspace distance $4$ or $2k-2$ were constructed in \cite{penddot2,penddot}. In \cite{Linkage}, Gluesing-Luerssen and
Troha presented the linkage construction by linking two CDCs with small length to obtain CDCs with improved  lower bounds. Xu and Chen \cite{Parallel} established the so-called parallel linkage construction by modifying the linkage construction and obtained new lower bounds for $A_q(2k,d,\{k\})$ when $k \geq d$, which was the first significant improvement
 of the linkage construction. Later, He \cite{TwoParal} constructed CDCs using two parallel linkage constructions which generalized the result of \cite{Parallel}. There are primarily two methods for constructing CDCs based on the parallel linkage construction in \cite{TwoParal}. The first method is to combine the parallel linkage construction with the multilevel construction to
construct larger CDCs  \cite{ParControl,MulLink,ParMul,ParLinkMul}. The second method is the block
inserting construction,  which flexibly inserts CDCs constructed from matrix blocks of small CDCs and rank-metric codes into  the CDCs obtained from
the parallel linkage construction \cite{BlockInsert1,BlockInsert2,Insert0,GenInsert}. For more constructions of CDCs and lower bounds for CDCs, one can refer to \cite{Coset,GerLinkage,LinkMul}. Recently, the most notable work is the mixed dimension construction \cite{MixDD}, which is based on mixed dimension/distance subspace codes and  rank-metric codes. Subsequently, He et al. \cite{ParllelMixed} applied the parallel construction to the mixed dimension construction, and obtained some new lower bounds for CDCs.

By  defining inverse identifying vectors
and Ferrers diagram rank-metric codes with given ranks, Liu and Ji \cite{DoubleMul} proposed the inverse multilevel construction, which can be
seen as a generalization of the parallel construction. By combining the multilevel construction and the
inverse multilevel construction, they introduced the double multilevel construction, which helps to reduce the complexity of selecting identifying vectors. Yu et al. \cite{Bilateral} generalized the concepts of identifying and inverse identifying vectors by introducing bilateral identifying vectors. By using bilateral Ferrers diagram rank-metric codes, they proposed the bilateral multilevel construction, which can effectively alleviate the problem of selecting identifying vectors. The generalized bilateral multilevel construction \cite{GBilateral} generalizes the bilateral multilevel construction and can be combined with the double multilevel construction to produce new CDCs.

This paper is devoted to constructing CDCs of larger sizes. Inspired by the ideas in \cite{ParllelMixed,GBilateral}, we propose a generalized bilateral multilevel construction based on the parallel mixed dimension construction and obtain some new lower bounds for CDCs. Section \ref{Sec-Pre} revisits some fundamental concepts and known results on rank-metric codes, Ferrers diagram rank-metric codes, the double multilevel construction, the generalized bilateral multilevel construction, and the mixed dimension construction. Section \ref{Sec-main} presents criteria for selecting bilateral identifying vectors and shows how to combine the generalized bilateral multilevel construction with the parallel mixed dimension construction in Theorems \ref{T-BPM-CDC} and \ref{T-SPar-CDC}.   By applying Theorem \ref{T-BPM-CDC}, we choose some bilateral identifying vectors to construct CDCs. Similarly, by applying Theorem \ref{T-SPar-CDC}, we propose a simpler method for selecting bilateral identifying vectors to construct CDCs. In Section \ref{Sec-Result}, we obtain some new lower bounds for CDCs and compare them with
the previous works. In Section \ref{Sec-Conclusion}, we conclude the paper.

\section{Preliminaries} \label{Sec-Pre}

In this section, we review some basic definitions and known results which are necessary for obtaining our results.
\subsection{Rank-metric codes}
Let $\BF_q^{m \times n}$ denote the set of all $m \times n$ matrices over $\BF_q$. For any two matrices $\bA$ and $\bB$ in $\BF_q^{m \times n}$, the rank distance $d_R(\bA, \bB)$ is defined as $\rank(\bA-\bB)$. An $(m \times n,d)_q$ rank-metric code (RMC) $\cC$ is a subset of $\BF_q^{m \times n}$ with a minimum rank distance
\[
d= \min_{\bA,\bB \in \cC, \bA \neq 
\bB} \{ d_R (\bA,\bB) \}.
\]
Analogous to Hamming metric codes, the size of a rank-metric code $\cC$ has a Singleton-like upper bound, which is given by \[\Delta(m,n,d)_q := q^{\max\{m,n\}(\min\{m,n\}-d+1)}.\] When the cardinality of $\cC$ achieves this bound, it is called a maximum rank distance (MRD) code. Furthermore, if $\cC$ forms an $\BF_q$-linear subspace of $\BF_q^{m \times n}$, then it is referred to as a linear code and is denoted by an $[m \times n,d]_q$ RMC. Linear MRD codes exist for all feasible parameters \cite{DELSARTE,1985Theory,Array}. The rank distribution of a linear MRD code can be determined from its parameters.

\begin{theorem}[Delsarte Theorem \cite{DELSARTE}]\label{T-rank-distribution}
Let $m, n, d$, and $i$ be positive integers with $m  \geq d$, $n \geq d$.  Then the number of matrices with rank $i$ in an $[m \times n, d]_q$ MRD code is given by
\[D(m,n,d,i)_q := \begin{bmatrix}
	\min\{m,n\} \\
	i
\end{bmatrix}_{q}
\sum_{j = 0}^{i - d}(-1)^{j} q^{
\frac{j(j-1)}{2}
}
\begin{bmatrix}
	i \\
	j
\end{bmatrix}_{q}
(q^{\max\{m,n\}(i - d - j + 1)} - 1)
\]
for $d \leq i \leq \min\{m,n\}.$
\end{theorem}

If the rank of each codeword in an $(m \times n, d)_q$ RMC is at most $r$, then it is called a \emph{rank-restricted rank-metric code} (RRMC), and is denoted by an $(m \times n, d;r)_q$ RRMC. Let $\cD$ be the set of elements from an $[m \times n,d]_q$ MRD code such that the rank of each matrix is at most $r$. Then $\cD$ is an $(m \times n,d;r)_q$ RRMC with cardinality
\[
 \Delta(m,n,d;r)_q := 1+\sum_{i=d}^r D(m,n,d,i)_q.   
\]

%Let $d$ and $d_1$ be two positive integers with $d_1 > d$. We can assume that an $[m \times n, d]_q$ MRD code contains an $[m \times n, d_1]_q$ MRD code as a subcode. The following construction uses a linear MRD code to obtain disjoint MRD subcodes.

%\begin{lemma}[Subcode Construction \cite{Subcode} ]\label{L-Subcode}
%Let $\cM$ be an $[m \times n, d]_q$ MRD code containing an $[m \times n, d_1]_q$ MRD subcode $\cM_1$, where $d_1 > d.$  Let $s = \frac{\Delta(m, n, d)_q}{\Delta(m, n, d_1)_q}$.
%Then there exist $s$ disjoint MRD subcodes of $\cM$  fulfilling  the following conditions:
%\begin{itemize}
%\item[(1)] $\cM_i$ is an $(m \times n, d_1)_q$ MRD code, for all $1 \leq i \leq s$;
%\item[(2)] for all $1 \leq i < i' \leq s$, it holds that  $\bM \neq \bM'$ and $\rank(\bM - \bM') \geq d$ for all $\bM \in \cM_i$, $\bM' \in \cM_{i'}$.
%\end{itemize}
%Moreover, $\cM_1$ is the unique linear MRD code among these $s$ subcodes.
%\end{lemma}

\subsection{Ferrers diagram rank-metric codes}
Etzion and Silberstein \cite{Multilevel} proposed the multilevel construction of CDCs by introducing Ferrers
diagram rank-metric codes, which generalizes the lifted MRD code construction.

Given positive integers $m$ and $n$, an $m \times n$ \emph{Ferrers diagram} $\cF$ is an $m \times n$ array consisting of dots and empty cells, such that
\begin{itemize}
    \item[(1)] all the dots are shifted to the right of the diagram;
    \item[(2)] the number of dots in each row is at most the number of dots in the previous row;
    \item[(3)] the top row contains $n$ dots and the rightmost column contains $m$ dots.
\end{itemize}

Denote $\gamma_i$ as the number of dots in the $i$-th column of the Ferrers diagram $\cF$, where $1 \leq i \leq n$. Given positive integers $m$ and $n$, and $1 \leq \gamma_1 \leq \gamma_2 \leq \cdots \leq \gamma_{n} = m$, there is a unique Ferrers diagram $\cF$ of size $m \times n$, where each $i$-th column has $\gamma_i$ dots for any $1 \leq i \leq n$. In this instance, we write $\cF = [\gamma_1, \gamma_2, \dots, \gamma_{n}]$, inverse Ferrers diagram $\bar{\cF}=[\gamma_n, \gamma_{n-1}, \dots, \gamma_1]$, and transposed Ferrers diagram as 
$\cF^\mathrm{T} =[\rho_{m}, \rho_{m-1}, \dots, \rho_1]$, where $\rho_i$ represents the number of dots in the $i$-th row of $\cF$, for $1 \leq i \leq m$. Note that an $m \times n$ Ferrers diagram $\cF$ is called  full if it has $mn$ dots. 
%For simplicity, we denote the $m \times n$ full Ferrers diagram by $\cD_{m \times n}$.

\begin{example}
    Let $\cF=[1,1,2,3]$. Then 
    \[\cF=\begin{array}{cccc}
\bullet & \bullet &\bullet &\bullet\\
 & & \bullet &\bullet\\
 &&&\bullet
    \end{array}, \bar{\cF}=\begin{array}{cccc}
       \bullet  & \bullet & \bullet &\bullet  \\
        \bullet & \bullet & & \\
        \bullet&&&
    \end{array},  
 \cF^{\mathrm{T}}=\begin{array}{ccc}
    \bullet&\bullet&\bullet\\
    &\bullet&\bullet\\
    &&\bullet\\
    &&\bullet
    \end{array}.\] 
\end{example}

\begin{definition}\label{D-FRcode}
Given an $m \times n$ Ferrers diagram $\cF$, an $(\cF, d)_q$ \emph{Ferrers diagram rank-metric code} (FD code) is an $(m \times n, d)_q$ RMC  in which each  matrix $\bM$ has the form $\cF$, that is, all entries of $\bM$ that do not correspond to dots in  $\cF$ are zeros. Furthermore, an $(\cF, d)_q$ FD code with $M$ codewords is denoted by an $(\cF,M, d)_q$ FD code.
\end{definition}
\begin{remark}
 If $\cF$ is a full $m \times n$ diagram, then its corresponding FD code is simply a classical rank-metric code. 
\end{remark}

By replacing the Ferrers diagram
$\cF$ in Definition \ref{D-FRcode} with inverse Ferrers diagram $\bar{\cF}$, we can similarly define an $(\bar{\cF}, d)_q$ inverse FD code and an $(\bar{\cF}, M,  d)_q$ inverse FD code. Note that if there exists an  $(\cF, M,  d)_q$ FD code, then so does an $(\bar{\cF}, M, d)_q$ inverse FD code and an $(\cF^\mathrm{T}, M, d)_q$ FD code. Etzion and Silberstein \cite{Multilevel} established a Singleton-like
upper bound on FD codes as follows.

\begin{lemma}[\cite{Multilevel}]\label{L-Singleton}
Given an $(\cF, M, d)_q$ Ferrers diagram rank-metric code, let $v_i$ be the number of dots in $\cF$ after removing the top $i$ rows and the rightmost $d - 1 - i$ columns for $0 \leq i \leq d - 1$. Then for any $(\cF,M,d)_q$ FD code, $M \leq q^{\min_i\{v_i\}}$.
\end{lemma}

FD codes attaining the upper bound in Lemma \ref{L-Singleton} are called optimal. Constructions
of optimal FD codes can be found in \cite{opt-F1,penddot2,Opt-FDRMC,opt-F3,opt-F4,penddot,opt-F6}. The
following lemma provides a construction for optimal FD codes based on subcodes of MRD codes.

\begin{lemma}[\cite{Opt-FDRMC}]\label{T-opt-FDRMC}
Assume $\cF$ is an $m \times n$ $(m \geq n)$ Ferrers diagram,  and each of the rightmost $d - 1$ columns has at least $n$ dots. Let $\gamma_i$ be the number of dots in the $i$-th column of $\cF$. Then there exists an optimal $(\cF, q^{\sum_{i = 1}^{n - d+1}\gamma_{i}}, d)_q$ FD code for any prime power $q$.
\end{lemma}

\begin{definition}\label{D-RFD}
    Let $\cF$ be an $m \times n$ Ferrers diagram.
For nonnegative integers $d$ and $r$, an $(\cF, d;r)_{q}$ \emph{rank-restricted Ferrers diagram rank-metric code} (RFD code) is an $(\cF, d)_q$ FD code such that the rank of each codeword is at most $r$.
\end{definition}

By replacing the Ferrers diagram
$\cF$ in Definition \ref{D-RFD} with inverse Ferrers diagram $\bar{\cF}$, we can similarly define an $(\bar{\cF}, d; r)_q$  RFD code. More constructions on  RFD codes can be found in \cite{DoubleMul}.

\subsection{Echelon Ferrers form and inverse echelon Ferrers form}

Let  $U$ be a subspace in $\mathcal{G}_{q}(n,k)$. A \emph{generator matrix} of $U$ is a $k \times n$ matrix whose  rowspace forms $U$. A $k \times n$ matrix with rank $k$ is in \emph{reduced row (row inverse) echelon form} if the following conditions are satisfied:
\begin{itemize}
    \item[(1)] the leading coefficient of a row is always to the right (left) of the leading coefficient of the previous row;
    \item[(2)] all leading coefficients are ones;
    \item[(3)] every leading coefficient is the only nonzero entry in its column.
\end{itemize}
There is exactly one generator matrix of $U$ in 
reduced row echelon form, and it will be denoted by $E(U)$. Similarly, there is exactly one generator matrix of $U$ in reduced row inverse echelon form, and it will be denoted by $\bar{E}(U)$.

The identifying and inverse identifying vectors of a subspace are defined as follows. For more details, one can refer to \cite{Multilevel,DoubleMul}.

\begin{definition}
 Let  $U \in \mathcal{G}_{q}(n,k)$. The identifying vector of $U$, denoted by $i(U)$, is a binary vector of length $n$ and Hamming weight $k$, where the ones in $i(U)$ are in the positions where $E(U)$  has  pivot columns.  
\end{definition}

\begin{definition}
 Let  $U \in \mathcal{G}_{q}(n,k)$. The inverse identifying vector of $U$, denoted by $\bar{i}(U)$, is a binary vector of length $n$ and Hamming weight $k$, where the ones in $\bar{i}(U)$ are in the positions where $\bar{E}(U)$  has  pivot columns.  
\end{definition}
\begin{example}\label{E(U)}
    Consider a subspace $U \in \cG_2(6,3)$ with reduced row echelon form and reduced row inverse echelon form
    \begin{align*}
 E(U)=\begin{pmatrix}
1&0&0&0&1&0\\
0&0&1&0&1&1\\
0&0&0&1&0&1
        \end{pmatrix}, 
\bar{E}(U)=\begin{pmatrix}
1&0&1&0&0&1\\
1&0&0&0&1&0\\
1&0&1&1&0&0
\end{pmatrix}.
    \end{align*}
    The identifying vector and inverse identifying vector of $U$ are $i(U)=(101100)$ and $\bar{i}(U)=(000111)$, respectively.
\end{example}

The \emph{echelon Ferrers form} of a vector $\mathbf{v}$ of length $n$ and  weight $k$, denoted by $EF(\mathbf{v})$, is the $k \times n$ matrix in reduced row echelon form (in RREF) with leading entries (of rows) in the columns indexed by the nonzero entries of $\mathbf{v}$ and ``$\bullet$" in  all entries which do not have terminals zeros or ones. All dots in $EF(\mathbf{v})$ form a Ferrers diagram, denoted by $\cF_{\bv}$, and we call it the Ferrers diagram of $\bv$. The inverse echelon Ferrers form $\overline{EF}(\bar{\mathbf{v}})$ and the inverse Ferrers diagram $\bar{\cF}_{\bar{\bv}}$ of an inverse identifying vector $\bar{\bv}$ can be defined similarly.
\begin{example}
    For $i(U)=(101100)$ and $\bar{i}(U)=(000111)$ in Example \ref{E(U)}, the echelon Ferrers form $EF(i(U))$ of $i(U)$ and the inverse echelon Ferrers form $\overline{EF}(\bar{i}(U))$ of $\bar{i}(U)$ are the following two  matrices:
    \begin{align*}
EF(i(U))=\begin{pmatrix}
1&\bullet&0&0&\bullet&\bullet\\
0&0&1&0&\bullet&\bullet\\
0&0&0&1&\bullet&\bullet
        \end{pmatrix}, 
\overline{EF}(\bar{i}(U))=\begin{pmatrix}
\bullet&\bullet&\bullet&0&0&1\\
\bullet&\bullet&\bullet&0&1&0\\
\bullet&\bullet&\bullet&1&0&0
\end{pmatrix}.
    \end{align*}
    
The Ferrers diagram $\cF_{i(U)}$ of $i(U)$ and the inverse Ferrers diagram $\bar{\cF}_{\bar{i}(U)}$ of $\bar{i}(U)$ are
\begin{align*}
 \cF_{i(U)}=\begin{matrix}
\bullet&\bullet&\bullet\\
&\bullet&\bullet\\
&\bullet&\bullet
        \end{matrix}, 
\bar{\cF}_{\bar{i}(U)}=\begin{matrix}
\bullet&\bullet&\bullet\\
\bullet&\bullet&\bullet\\
\bullet&\bullet&\bullet
\end{matrix},
\end{align*}
respectively.
\end{example}

\subsection{Double multilevel construction}

This subsection is devoted to presenting an effective construction for CDCs, called the double multilevel construction,  which combines the multilevel construction and the
inverse multilevel construction. The subsequent lemmas are useful in the multilevel
construction and inverse multilevel construction. Let $d_H(\cdot)$ denote the Hamming metric.

\begin{lemma}[\cite{Multilevel}]\label{L-S-H}
Let $U$ and $V$ be two subspaces in 
$\mathcal{G}_{q}(n,k)$. Then \[d_{S}(U,V) \geq d_{H}(i(U), i(V)).\] 
\end{lemma}

\begin{lemma}[\cite{DoubleMul}]\label{L-inv-L-S-H}
Let $U$ and $V$ be two subspaces in $\mathcal{G}_{q}(n,k)$. Then  \[ d_{S}(U,V) \geq d_{H}(\bar{i}(U), \bar{i}(V)).\]
\end{lemma}

\begin{lemma}[\cite{Multilevel}]
    Let $U, V \in \cG_q(n,k)$. If $i(U) = i(V)$, then 
    \[d_S(U,V)=2d_R(C_{E(U)}, C_{E(V)}),\]
    where $C_{E(U)}$  and  $C_{E(V)}$ denote the submatrices of $E(U)$ and $E(V)$ without the pivot columns, respectively.
\end{lemma}

\begin{lemma}[\cite{DoubleMul}]
   Let $U, V \in \cG_q(n,k)$. If $\bar{i}(U) = \bar{i}(V)$, then 
    \[d_S(U,V)=2d_R(C_{\bar{E}(U)}, C_{\bar{E}(V)}),\]
    where $C_{\bar{E}(U)}$  and  $C_{\bar{E}(V)}$ denote the submatrices of $\bar{E}(U)$ and $\bar{E}(V)$ without the pivot columns, respectively. 
\end{lemma}

Lifting construction is an important technique for constructing CDCs. Let $\bv$ be a binary vector of length $n$ and weight $k$. Assume $\cC_{\cF_{\bv}}$ is an $(\cF_{\bv},d)_q$ FD code. Fill the $\cF_{\bv}$ in $EF(\bv)$ with the codewords of $\cC_{\cF_{\bv}}$, and then the rowspaces of these elements constitute the lifted CDC of $\cC_{\cF_{\bv}}$, denoted by $\sL(\cC_{\cF_{\bv}})$. The lifting technique is also applicable to inverse FD codes.
 Using the lifting technique, Etzion and Silberstein \cite{Multilevel} initially introduced the multilevel construction, and subsequently, Liu et al. \cite{DoubleMul} presented the inverse multilevel construction.

\begin{construction}[Multilevel construction 
\cite{Multilevel}]\label{C-Mul}
Let $\cS$ be a binary constant weight code of length $n$, weight $k$ and minimum Hamming distance $2\delta$. If there exists an $(\cF_{\bv},\delta)_q$ FD code $\cC_{\cF_{\bv}}$ for each $\bv \in \cS$, then its lifted code $\sL(\cC_{\cF_{\bv}})$ is an $(n,2\delta,\{k\})_q$ CDC. Moreover, the union of these lifted codes $\cup_{\bv \in \cS} \sL(\cC_{\cF_{\bv}})$ forms an $(n, \sum_{\bv\in \cS} |\cC_{\cF_{\bv}}| , 2\delta, \{k\})_q$ CDC.  
\end{construction}

\begin{construction}[Inverse multilevel construction 
\cite{DoubleMul}]\label{C-InvMul}
Let $\bar{\cS}$ be a binary constant weight code of length $n$, weight $k$ and minimum Hamming distance $2\delta$. Given a non-negative integer $s_{\bar{\bv}}$, if there exists an $(\bar{\cF}_{\bar{\bv}},\delta; s_{\bar{\bv}})_q$ RFD code $\cC_{\bar{\cF}_{\bar{\bv}}}$ for each $\bar{\bv} \in \bar{\cS}$, then its lifted code $\sL(\cC_{\bar{\cF}_{\bar{\bv}}})$ is an $(n,2\delta,\{k\})_q$ CDC. Moreover, the union of these lifted codes $\cup_{\bar{\bv} \in \bar{\cS}}\sL(\cC_{\bar{\cF}_{\bar{\bv}}})$ forms an $(n, \sum_{\bar{\bv} \in \bar{\cS}} |\cC_{\bar{\cF}_{\bar{\bv}}}|, 2\delta, \{k\})_q$ CDC.
\end{construction}

Based on the multilevel construction and the
inverse multilevel construction, the double multilevel construction was proposed in \cite{DoubleMul}.

\begin{construction}[Double multilevel construction \cite{DoubleMul}]\label{DoubMul}
 Let $n \geq 2k \geq 2\delta$. Take an $(n,M_1, 2\delta,\{k\})_q$ CDC $\cC_1$ from Construction \ref{C-Mul} with $\cS$ as the set of its identifying vectors, and take an $(n, M_2, 2\delta,\{k\})_q$ CDC $\cC_2$ from Construction \ref{C-InvMul} with $\bar{\cS}$ as the set of its inverse identifying vectors. Let $s_{\bar{\bv}}$ be an integer defined in Construction \ref{C-InvMul}. If $d_H(\bv,\bar{\bv}) \geq 2(s_{\bar{\bv}}+\delta)$ for any $\bv \in \cS$ and $\bar{\bv} \in \bar{\cS}$, then $\cC_1 \cup \cC_2$ is an $(n,M_1+M_2,2\delta,\{k\})_q$ CDC.
\end{construction}

\subsection{Generalized bilateral multilevel construction}
Let $n$, $n_1$, and $n_2$ be positive integers with $n \geq n_1 + n_2$. A \emph{bilateral identifying vector} $\tbv$ with length $n$ is a binary vector satisfying the following conditions:
\begin{itemize}
    \item[(1)] the subvector $\bv_1$, consisting of the first $n_1$ coordinates of $\tbv$, is an identifying vector;
    \item[(2)] the subvector $\bar{\bv}_2$, consisting of the last $n_2$ coordinates of $\tbv$, is an inverse identifying vector;
    \item[(3)] the subvector $\tbv_3$, consisting of the remaining $n-n_1-n_2$ coordinates of $\tbv$, is a zero vector.
\end{itemize}
The ordered triple $(n_1, n-n_1-n_2, n_2)$ is called the type of this bilateral identifying vector. To distinguish bilateral identifying vectors from identifying vectors and inverse identifying vectors, we denote a bilateral identifying vector by
\[
\tbv=(\overbrace{\bv_1}^{n_1} | \overbrace{\tbv_3}^{n-n_1-n_2} | \overbrace{\bar{\bv}_2}^{n_2} )=(x_1, \dots, x_{n_1}, \tilde{x}_{n_1+1}, \dots,\tilde{x}_{n-n_2}, \bar{x}_{n-n_2+1}, \dots, \bar{x}_{n}).
\]

The generalized bilateral echelon Ferrers form of a bilateral identifying vector is defined
as follows.

\begin{definition}[\cite{GBilateral}]
    Let $\tbv=(\bv_1|\tilde{\bv}_3|\bar{\bv}_2)$ be a bilateral identifying vector of type $(n_1,n-n_1-n_2, n_2)$, where $\wt(\bv_1)=a_1$, $\wt(\bar{\bv}_2)=a_2$, and $a_1+a_2=k$. Then the \emph{generalized bilateral echelon Ferrers form} $\widetilde{EF}(\tbv)$ of $\tbv$ is a $k \times n$ matrix such that
    \begin{itemize}
        \item[(1)] the $a_1 \times n_1$ submatrix in the upper left-hand corner of $\widetilde{EF}(\tbv)$ is $EF(\bv_1)$;
        \item[(2)] the $a_2 \times n_2$ submatrix in the lower right-hand corner of $\widetilde{EF}(\tbv)$ is $\overline{EF}(\bar{\bv}_2)$;
        \item[(3)] the $a_2 \times n_1$ submatrix in the lower left-hand corner of $\widetilde{EF}(\tbv)$ is a zero matrix;
        \item[(4)] in the $a_1 \times n_2$ submatrix in the upper right-hand corner of $\widetilde{EF}(\tbv)$, the entries in the columns of $\widetilde{EF}(\tbv)$ that contain ones of $\overline{EF}(\bar{\bv}_2)$ are zeros and other entries are dots;
        \item[(5)] the remaining $n-n_1-n_2$ columns of $\widetilde{EF}(\tbv)$ form a $k \times (n-n_1-n_2)$ full Ferrers diagram.
    \end{itemize}
\end{definition}

Consider a special bilateral identifying vector such
that all $1$’s are consecutive in the first $n_1$ coordinates, and all $\bar{1}$’s are  consecutive in the last $n_2$ coordinates. We use this vector as an example to illustrate its generalized bilateral echelon Ferrers form, which is shown as follows.
Take a bilateral identifying vector
\[
\tbv=(\overbrace{\underbrace{0\cdots0}_{u_1} \underbrace{1\cdots1}_{a_1} \underbrace{0\cdots0}_{u_2}}^{n_1} \overbrace{\tilde{0}\cdots\tilde{0}}^{n-n_1-n_2} \overbrace{\underbrace{\bar{0}\cdots\bar{0}}_{u_3} \underbrace{\bar{1}\cdots\bar{1}}_{a_2} \underbrace{\bar{0}\cdots\bar{0}}_{u_4}}^{n_2}), 
\]
where $n_1 +n_2 \leq n$, $u_1$, $u_2$, $u_3$, $u_4$ are nonnegative integers and $a_1+a_2=k$. Denote $\bO_{a \times b}$ as the all zero matrix of size $a\times b$, and we will ignore the size if the context is clear.
Denote  $\bI_s$ as the $s \times s$ identity matrix, and $\bar{\bI}_s$ as the $s \times s$ matrix 
\begin{align*}
 \begin{pmatrix}
0 & 0 & \cdots & 0 & 1\\
0 & 0 & \cdots &1 &0\\
\vdots & \vdots &\ddots & \vdots & \vdots \\
1 & 0 & \dots & 0 &0 
\end{pmatrix}.   
\end{align*}
Then the generalized bilateral echelon Ferrers form of $\tbv$ is 
\begin{align*}
\widetilde{EF}(\tbv)=\begin{pmatrix}
  \bO & \bI_{a_1} & \cF_1 & \cF_2 & \cF_3 & \bO & \cF_4 \\
  \bO & \bO & \bO & \cF_5 & \bar{\cF}_6 & \bar{\bI}_{a_2} & \bO
\end{pmatrix},  
\end{align*}
where $\cF_1$, $\cF_2$, $\cF_3$, $\cF_4$, $\cF_5$, and $\bar{\cF}_6$ are full Ferrers diagrams, with sizes  $a_1 \times u_2$, $a_1 \times (n-n_1-n_2)$, $a_1 \times u_3$,  $a_1 \times u_4$, $a_2 \times (n-n_1-n_2)$, and $a_2 \times u_3$, respectively.

The generalized bilateral Ferrers diagram of a bilateral identifying vector is defined as follows.

\begin{definition}[\cite{GBilateral}]
 Let $\tbv=(\bv_1|\tilde{\bv}_3|\bar{\bv}_2)$ be a bilateral identifying vector of type $(n_1,n-n_1-n_2, n_2)$, where $\wt(\bv_1)=a_1$, $\wt(\bar{\bv}_2)=a_2$, and $a_1+a_2=k$. The \emph{generalized bilateral Ferrers diagram} $\tcF_{\tbv}$ of $\tbv$ is formed by all dots in $\widetilde{EF}(\tbv)$. More specifically,
\begin{align*}
    \tcF_{\tbv}= \begin{matrix}
        \cF_{\bv_1} & \cF_1 & \cF_2 \\
        \ & \cF_3 & \bar{\cF}_{\obv_2}
    \end{matrix},
\end{align*}
where $\cF_1$ is an  $a_1 \times (n-n_1-n_2)$ full Ferrers diagram, $\cF_2$ is an $a_1 \times (n_2-a_2)$ full Ferrers diagram, and $\cF_3$ is an $a_2 \times (n-n_1-n_2)$ full Ferrers diagram.
\end{definition}

Consider a bilateral identifying vector $\tbv=(\bv_1 | \tbv_3 | \obv_2)$ of type $(n_1,n-n_1-n_2,n_2)$ with $\wt(\bv_1) = a_1$, $\wt(\obv_2)= a_2$,  and $a_1+a_2=k$. For any matrix $\bM$ with form $\tcF_{\tbv}$, define $\phi_{\tbv}(\bM)$ as the $a_1 \times (n_2-a_2)$ submatrix located in the upper right-hand corner of $\bM$. Clearly, the mapping $\phi_{\tbv}$ is related to $\tbv$.

\begin{example}
    Let $\tbv=(110100 \tilde{0}\tilde{0} \bar{0}\bar{0}\bar{1}\bar{0}\bar{1}\bar{1}\bar{0})$ be a bilateral identifying vector. The generalized bilateral echelon Ferrers form of $\tbv$ is
\begin{align*}
\widetilde{EF}(\tbv)=
\left(\begin{array}{ccccccccccccccc}     
1 & 0 & \bullet & 0 & \bullet  & \bullet & \bullet & \bullet & \bullet &\bullet & 0 & \bullet &0&0 & \bullet \\ 0 & 1 & \bullet & 0 & \bullet  & \bullet & \bullet & \bullet &\bullet& \bullet & 0 & \bullet &0&0& \bullet\\ 0 &0&0&1& \bullet&\bullet&\bullet&\bullet&\bullet&\bullet&0&\bullet&0&0&\bullet \\ 0&0&0&0&0&0&\bullet&\bullet &\bullet&\bullet &0 & \bullet &0 &1&0\\
0&0&0&0&0&0&\bullet &\bullet&\bullet&\bullet&0&\bullet&1&0&0\\
0&0&0&0&0&0&\bullet &\bullet&\bullet&\bullet&1&0&0&0&0
\end{array}\right),
\end{align*}
and the generalized bilateral Ferrers diagram $\tcF_{\tbv}$ is
\begin{align*}
\tcF_{\tbv}=\begin{array}{ccccccccc}
    \bullet&\bullet & \bullet &\bullet&\bullet&\bullet&\bullet&\bullet&\bullet \\
     \bullet&\bullet & \bullet &\bullet&\bullet&\bullet&\bullet&\bullet&\bullet \\ 
     &\bullet&\bullet & \bullet &\bullet&\bullet&\bullet&\bullet&\bullet\\
     & & & \bullet &\bullet&\bullet&\bullet&\bullet&\\
     & & & \bullet &\bullet&\bullet&\bullet&\bullet&\\
     & & & \bullet &\bullet&\bullet&\bullet&&
\end{array}.
\end{align*}
Take a matrix 
\begin{align*}
\bM=\left(
\begin{array}{ccccccccc}
    1 & 1 &0&1&0&0&1&1&1 \\
     1 & 1 &1&0&1&0&1&0&1 \\ 
    0 &1 & 0 &1&1&1&0&1&0\\
     0& 0 & 0 &1&1&0&1&0&0\\
     0& 0 & 0 &1&1&0&0&1&0\\
     0& 0 & 0 &0&0&1&1&0&0
\end{array}\right)
\end{align*}
with form $\tcF_{\tbv}$, then $\phi_{\tbv}(\bM)=\left(\begin{array}{cccc}
   0 &1&1&1  \\
    0 & 1&0&1 \\
    1 & 0&1 &0
\end{array}\right)$.
\end{example}

\begin{definition}
    Let $\tcF$ be an $m \times n$ generalized bilateral Ferrers diagram.
     An $(\tcF, d)_{q}$ \emph{generalized bilateral Ferrers diagram rank-metric code} (GB-FD code) is an $(m \times n, d)_q$ RMC in which each codeword has form $\tcF$. An $(\tcF, d)_{q}$ GB-FD code with $M$ codewords is denoted by an $(\tcF,M, d)_{q}$ GB-FD code.
\end{definition}

The following proposition provides a lower bound on the size of special GB-FD codes, in which the rank of a specific submatrix of the codeword is restricted. Let $\sigma(\cdot)_{a,b}: \mathbb{F}_{q}^{m \times n} \rightarrow \mathbb{F}_{q}^{a \times b}$ be the map from a matrix $\bM$ to the $a \times b$ submatrix in the upper right-hand corner of $\bM$.

\begin{proposition}[\cite{GBilateral}]\label{P-GBFD-Par}
    Let $\delta$ be a positive integer, and let $k_i \geq \delta$, $l_i \geq \delta$ for $1 \leq i \leq 3$,  with $k_2 \geq k_1$ and  $l_2 \geq l_3$. Let $\Lambda_i = \max\{k_i,l_i\}$ and  $\Gamma_i=\min\{k_i,l_i\}$ for $i=1, 3$. Set $d_1 =\begin{cases}
    \lceil \frac{\delta}{2} 
    \rceil & ~\text{if}~\Lambda_1 \geq \Lambda_3;\\
    \lfloor \frac{\delta}{2} 
    \rfloor & \text{otherwise},
\end{cases}$ and $d_2=\delta-d_1$. Suppose 
\begin{align*}
\tcF=\begin{matrix}
    \cF_1 & \cF_2\\
    & \cF_3
\end{matrix}
\end{align*}
is a generalized bilateral Ferrers diagram, where $\cF_i$ is a $k_i \times l_i$ full Ferrers diagram for $i=1,2,3$. Then there exists an $(\tcF,\rho,\delta)_q$ GB-FD code $\cC_{\tcF}$ such that for any codeword $\bM \in \cC_{\tcF}$, $\rank(\sigma(\bM)_{k_2,l_2}) \leq r$, where
\[
\rho=\min\{q^{\Lambda_1 d_2}, q^{\Lambda_3 d_1}\} q^{\Lambda_1(\Gamma_1-\delta+1)+\Lambda_3(\Gamma_3 -\delta+1)} \Delta(k_2, l_2, \delta;r)_q.
\]
\end{proposition}

The following two lemmas are crucial for presenting the generalized bilateral
multilevel construction. Let $\rs(\bU)$ denote the subspace of $\BF_q^n$ spanned by the rows of a matrix $\bU \in \BF_q^{k \times n}$.

\begin{lemma}[\cite{GBilateral}]\label{G1}
Let $\tbu$ and $\tbv$ be two bilateral identifying vectors having the same type $(n_1,n-n_1-n_2,n_2)$ with $n_1+n_2 \leq n$ and weight $k$.   Assume $\bU$ (resp. $\bV$) is a matrix in generalized bilateral echelon form $\widetilde{EF}(\tbu)$ (resp. $\widetilde{EF}(\tbv)$). Let $U=\rs(\bU)$ and $V=\rs(\bV)$, then 
     \[d_S(U,V) \geq d_H(\tbu,\tbv).\]
\end{lemma}

\begin{lemma}[\cite{GBilateral}]\label{G2}
   Let $\tbv=(\bv_1|\tilde{\bv}_3|\bar{\bv}_2)$ be a bilateral identifying vector of type $(n_1,n-n_1-n_2, n_2)$ with $n_1+n_2 \leq n$ and weight $k$. Assume $\bU$ and $\bV$ are two matrices in generalized bilateral echelon form $\widetilde{EF}(\tbv)$.  Let $U=\rs(\bU)$ and $V=\rs(\bV)$. Let $\tilde{C}_{\bU}$ (resp. $\tilde{C}_{\bV}$)  denote the submatrix of $\bU$ (resp. $\bV$) without the columns of pivots in $EF(\bv_1)$ and $\overline{EF}(\bar{\bv}_2)$. Then
   \[d_S(U,V) = 2d_R(\tilde{C}_{\bU},\tilde{C}_{\bV}).\]
\end{lemma}

The generalized bilateral multilevel construction is based on Lemmas \ref{G1} and \ref{G2}.
\begin{construction}[Generalized bilateral multilevel construction \cite{GBilateral}]\label{Const-GBMC}
Let $\tilde{\cS}$ be a set of bilateral identifying vectors, which is a binary constant-weight code of length $n$, weight $k$,  and minimum Hamming distance $2\delta$. Assume that all bilateral identifying vectors in $\tcS$ have the same type. If there exists an $(\tcF_{\tbv},\delta)_q$ GB-FD code $\cC_{\tcF_{\tbv}}$ for each $\tbv \in \tilde{\cS}$, then the lifted code $\sL(\cC_{\tcF_{\tbv}})$ is an $(n,2\delta,\{k\})_q$ CDC. 
Moreover, the union of these lifted codes $\cup_{\tbv \in \tilde{\cS}} \sL(\cC_{\tcF_{\tbv}})$ forms an $(n, \sum_{\tbv \in \tilde{\cS}} |\cC_{\tcF_{\tbv}}|, 2\delta, \{k\})_q$ CDC.
\end{construction}

\subsection{Mixed dimension construction}

Lao et al. \cite{MixDD} introduced the mixed dimension construction based on small RMCs and a new family of codes called mixed dimension/distance subspace codes.

\begin{definition}
    Let $\cC \subseteq \cP_q(n)$, $d_1$ and $d_0$ be two positive integers with $d_1 \geq d_0$. For any two distinct subspaces $X,Y \in \cC$,
    \begin{itemize}
        \item[(1)] $d_S(X,Y) \geq d_1$ if $\dim(X)=\dim(Y)$,
        \item[(2)] $d_S(X,Y) \geq d_0$ if $\dim(X) \neq \dim(Y)$,
    \end{itemize}
then $\cC$ is an $(n,d_1,d_0)_q$ mixed dimension/distance subspace code (MDDC).

Let $T$ be a subset of $\{0,1, \dots,n\}$. We say that $\cC$ is an $(n,d_1,d_0,T)_q$ MDDC, if $\eta_i(\cC) >0$ for $i \in T$ and $\eta_i(\cC) = 0$ for $i \notin T$.
\end{definition}
\begin{example}\label{E-Alg}
    By using the $(8,4801,4,\{4\})_2$ CDC and the $(8,1326,4, \{3\})_2$ CDC constructed in \cite{EAlg}, we have an $(8,4,3,\{4,3\})_2$ MDDC $\cC$ given by Algorithm 1 in \cite{MixDD}, which has dimension distribution $\eta_4(\cC)=4801$ and $\eta_3(\cC)=327$.
\end{example}

The following notations will be used in the mixed dimension construction.
\begin{itemize}
\item $(\bA~|~\bB)$ or $(\bA~\bB)$  represents a matrix concatenated from $\bA$ and $\bB$ with compatible sizes and ambient fields, 
\item for a subset $T$ of the positive integers, denote the minimum value in the set $T$  by $    T^{\min}$, and define
\begin{align*}
    l_T:= \begin{cases}
        \min\{u-v:u,v \in T, u>v\} & \text{if}~|T|>1, \\
        $0$ & \text{if}~ |T|=1.
    \end{cases}
\end{align*}
\end{itemize}

The SC-representation of a mixed dimension code is defined as follows, and it is  helpful for describing the mixed dimension construction. 

\begin{definition} \label{D-SC}
A set of matrices $\cH(\cC) \subseteq \BF_q^{k \times n}$ is called  an SC-representation of a constant dimension code $\cC$ in $\cG_q(n,k)$ if the following conditions are satisfied:
\begin{itemize} 
\item[(1)] for each $\bH \in \cH(\cC)$, $\rank(\bH) = k$;
\item[(2)] for any two distinct elements $\bH, \bH' \in \cH(\cC)$, $\rs(\bH) \neq  \rs(\bH')$;
\item[(3)] $\cC = \{\rs(\bH) \mid \bH \in \cH(\cC)\}$.
\end{itemize}
Furthermore, let $\cC$ be an $(n, 2\delta, T)_q$ mixed dimension code such that $ \cC= \cup_{k\in T}\cC^{k}$ with $\cC^{k} \subseteq \cG_q(n,k)$. 
An SC-representation of the mixed dimension code $\cC$ is defined as $\cH(\cC) = \cup_{k \in T} \cH(\cC^{k})$.
\end{definition}

 \begin{theorem}[Mixed dimension construction \cite{MixDD}]\label{Mix}
    Let $n, n_1, n_2, k,$ and $ \delta$ be integers such that $n=n_1+n_2$, $n_1 \geq k$, $n_2 \geq k$, and $k \geq \delta \geq 2$.  Let $T_1 \subseteq [\delta,k]$ with $l_{T_1} <2 \delta$. Let $T_2 \subseteq [k+\delta-T_1^{\min},k]$ with $l_{T_2} < 2\delta$.   For $i = 1,2$, let $\cH_i = \cup_{r \in T_i}\{\bH \in \BF_q^{r \times n_i}: \rank(\bH)=r, \bH~\text{in}~\RREF \}$ be an SC-representation of  an $(n_i, 2\delta, 2\delta-l_{T_i}, T_i)_q$ MDDC $\cX_i$. Let $\cP_t$ be a $(k \times (n_2+t-k),\delta)_q$ MRD code for all $t \in T_1$ and $\cQ_s$ be a $(k \times (n_1+s-k), \delta; T_1^{\min} - \delta -(k-s))_q$ RRMC for all $s\in T_2$.
    Define $\cC_1 = \cup_{t \in T_1} \cC_1^{(t)}$, 
    \begin{align*}
        \cC_1^{(t)} = \left\{
        \rs\left(
        \begin{array}{c|cc}
        \begin{matrix}
            \bH_1 \\ \bO   
        \end{matrix}
        &\begin{matrix}
            \bO \\
            \bI_{k-t}
        \end{matrix}
            & \bP
        \end{array}
        \right) 
        \right\},
        \end{align*}
        where $\bH_1 \in \cH_1, \rank(\bH_1) = t,$ and $\bP \in \cP_t$.
        Define $\cC_2 = \cup_{s \in T_2} \cC_2^{(s)}$,
        \begin{align*}
        \cC_2^{(s)} =  \left\{\rs\left( 
        \begin{array}{cc|c}  
            \bQ &
      \begin{matrix}
                \bO \\
          \bI_{k-s} 
        \end{matrix}
        &
    \begin{matrix}
            \bH_2\\
            \bO
        \end{matrix}
        \end{array}
        \right)
        \right\},
    \end{align*}
    where $\bH_2 \in \cH_2, \rank(\bH_2)=s$, and $\bQ \in \cQ_s$.
    Then  $\cC_1 \cup \cC_2$ is an $(n, 2\delta, \{k\})_q$ CDC with size  $\sum_{t \in T_1} \eta_t(\cX_1)|\cP_t| + \sum_{s \in T_2} \eta_s(\cX_2)|\cQ_s|$.
\end{theorem}

\section{Main results}\label{Sec-main}

In this section, we first introduce the parallel mixed dimension construction. Later, 
we describe our two generalized bilateral multilevel constructions based on the parallel mixed dimension construction.

\begin{theorem}[\cite{ParllelMixed}]\label{T-ParalMixConstr}
    Let $n$, $n_1$, $n_2$, $k$, and $\delta$ be integers with $n=n_1+n_2,$ $k \geq 2\delta$,  and $ n_i \geq k$ for $i=1,2$. Let $T_1 \subseteq [\delta,k]$ with $l_{T_1} <2 \delta$. Let $T_2 \subseteq [k+\delta-T_1^{\min},k]$ with $l_{T_2} < 2\delta$. Let $n_3$ be a positive integer satisfying $n_1+k-T_1^{\min} \leq n-n_3+T_2^{\min}-k$. 
    Let $\cH_1 = \cup_{t \in T_1}\{\bH_1 \in \BF_q^{t \times n_1}: \rank(\bH_1)=t, \bH_1~\text{in}~\RREF \}$ be an SC-representation of an $(n_1, 2\delta, 2\delta-l_{T_1}, T_1)_q$ MDDC $\cX_1$.
    Let $\cH_3= \cup_{s\in T_2}\{\bH_3 \in \BF_q^{s \times n_3}: \rank(\bH_3)=s, \bH_3~\text{in}~ \RREF\}$ be an SC-representation of an $(n_3, 2\delta, 2\delta-l_{T_2}, T_2)_q$ MDDC $\cX_3$. Let $\cP_t$ be a $(k \times (n_2+t-k), \delta)_q$ MRD code for all $t \in T_1$ and $\cQ_s'$ be a $(k\times(n-n_3+s-k), \delta; k-\delta)_q$ RRMC for all $s \in T_2$. 
    Define $\cC_1 = \cup_{t \in T_1} \cC_1^{(t)}$, 
    \begin{align*}
        \cC_1^{(t)} = \left\{
        \rs\left(
        \begin{array}{c|cc}
        \begin{matrix}
            \bH_1 \\ \bO   
        \end{matrix}
        &\begin{matrix}
            \bO \\
            \bI_{k-t}
        \end{matrix}
            & \bP
        \end{array}
        \right) 
        \right\},
        \end{align*}
        where $\bH_1 \in \cH_1, \rank(\bH_1) = t,$ and $\bP \in \cP_t$.
    Define $\cC_3 = \cup_{s \in T_2} \cC_3^{(s)}$, 
    \[\cC_3^{(s)} =  \left\{\rs\left( 
        \begin{array}{cc|c}  
            \bQ' &
      \begin{matrix}
                \bO \\
          \bI_{k-s} 
        \end{matrix}
        &
    \begin{matrix}
            \bH_3\\
            \bO
        \end{matrix}
        \end{array}
        \right)
        \right\},\]
        where $\bH_3 \in \cH_3, \rank(\bH_3)=s$, and $\bQ' \in \cQ'_s$.
        Then $\cC_1 \cup \cC_3$ is an $(n, 2\delta,\{k\})_q$ CDC with size $\sum_{t \in T_1} \eta_t(\cX_1)|\cP_t| + \sum_{s \in T_2} \eta_s(\cX_3)|\cQ_s'|$.
\end{theorem}

Take $n_3=k$ and $T_2=\{k\}$ in Theorem \ref{T-ParalMixConstr}, then $\cH_3$ degenerates to $\{\bI_k\}$, and the subsequent corollary  follows directly.

\begin{corollary}[\cite{ParllelMixed}]\label{S-ParMixConstr}
Let $n, n_1, n_2, k$, and $\delta$ be integers with $n=n_1+n_2$, $n_1 \geq k$, $n_2 \geq k$, and $k \geq 2\delta \geq 4$. Let $T_1 \subseteq [\delta,k]$ with $l_{T_1} < 2\delta$. Assume that $n_1 + k- T_1^{\min} \leq n-k$. Let $\cH_1 = \cup_{t \in T_1} \{\bH \in \BF_q^{t \times n_1}: \rank(\bH)=t, \bH~\text{in}~\RREF\}$ be an SC-representation of an $(n_1, 2\delta, 2\delta-l_{T_1}, T_1)_q$ MDDC $\cX_1$. Let $\cQ$ be a $(k\times(n-k), \delta; k-\delta)_q$ RRMC and $\cP_t$ be a $(k\times (n_2+t-k), \delta)_q$ MRD code for all $t \in T_1$. Define $\cC_1 = \cup_{t\in T_1} \cC_1^{(t)}$, 
\[\cC_1^{(t)} = \left\{
        \rs\left(
        \begin{array}{c|cc}
        \begin{matrix}
            \bH_1 \\ \bO   
        \end{matrix}
        &\begin{matrix}
            \bO \\
            \bI_{k-t}
        \end{matrix}
            & \bP
        \end{array}
        \right)\right\},\]
where  $\bH_1 \in \cH_1, \rank(\bH_1) = t$, and $\bP \in \cP_t$.
Define $\cC_3' = \{\rs(\bQ \mid \bI_k)\}$, where $\bQ \in \cQ$. Then $\cC_1 \cup \cC_3'$ is an $(n, 2\delta, \{k\})_q$ CDC with size $\sum_{t \in T_1} \eta_t(\cX_1)|\cP_t|+|\cQ|$. 
\end{corollary}

We combine the generalized bilateral multilevel construction with the parallel mixed dimension construction to produce CDCs as follows. The Hamming weight of a vector is denoted as $\wt(\cdot)$.

\begin{theorem}\label{Our-T-GBPM}
    Use the same notation as in Theorem \ref{T-ParalMixConstr}. Let $\mu_1=n_1+k-T_1^{\min}$, $\mu_2=n_3+k-T_2^{\min}$, and $\mu_3=n-\mu_1-\mu_2$.
    Let $\tcS \subseteq \BF_2^n$ be a set of bilateral identifying vectors  with  Hamming weight $k$ and minimum Hamming distance $2\delta$. Assume that all bilateral identifying vectors in $\tcS$ have the same type $(\mu_1, \mu_3, \mu_2)$.
    Suppose that for any bilateral identifying vector $\tbv = (\bv_1 | \tbv_3 | \obv_2)$ in $\tcS$, it satisfies $\delta \leq \wt(\bv_1) \leq k-\delta$. If there exists an $(\tcF_{\tbv},\delta)_q$ GB-FD code $\cC_{\tcF_{\tbv}}$ such that $\rank(\phi_{\tbv}(\bM)) \leq \wt(\bv_1)-\delta$  for any $\tbv \in \tcS$ and $\bM \in \cC_{\tcF_{\tbv}}$, then $\tcC_1 = \cup_{\tbv \in \tcS} \sL(\cC_{\tcF_{\tbv}})$ is an $(n, 2\delta,\{k\})_q$ CDC. Moreover, $\cC_1 \cup \cC_3 \cup \tcC_1$ is an $(n, |\cC_1|+|\cC_3|+|\tcC_1|, 2\delta, \{k\})_q$ CDC.
\end{theorem}
\begin{proof}
    By Construction \ref{Const-GBMC} , $\tcC_1$ is an $(n,2\delta,\{k\})_q$ CDC. It remains to show that $d_S(U, V) \geq 2\delta$ for any $U \in \cC_1 \cup \cC_3$ and $V \in \tcC_1$. We analyze the distance in the following two cases.

\textbf{Case 1:} $U \in \cC_1$.

    Suppose that the identifying vector of $U$ is $\bu=(\overbrace{\bu_1}^{\mu_1} \mid \overbrace{\bu_3}^{\mu_3} \mid \overbrace{\bu_2}^{\mu_2})$ and the bilateral identifying vector of $V$ is $\tbv=(\overbrace{\bv_1}^{\mu_1} \mid \overbrace{\tbv_3}^{\mu_3} \mid \overbrace{\bar{\bv}_2}^{\mu_2})$. It is clear that $\wt(\bu_1)=k$ and $\wt(\bu_3|\bu_2) = 0$. Assume that $\bv$ is the identifying vector of $V$, then we have $\bv=(\overbrace{\bv_1}^{\mu_1} | \overbrace{\bv_3}^{\mu_3} | \overbrace{\bv_2}^{\mu_2})$ with $\wt(\bv_3|\bv_2)= k - \wt(\bv_1)$. By Lemma \ref{L-S-H}, we obtain that
    \begin{align*}
        d_S(U,V) &\geq d_H(\bu,\bv)\\
        &= d_H(\bu_1,\bv_1)+d_H(\bu_3|\bu_2,\bv_3|\bv_2)\\
        &\geq d_H(\bu_1,\bv_1)+|\wt(\bu_3|\bu_2)-\wt(\bv_3|\bv_2)|\\
        &\geq |\wt(\bu_1)-\wt(\bv_1)| + |(k-\wt(\bu_1))-(k-\wt(\bv_1))| \\& \geq 2 \cdot [\wt(\bu_1)-\wt(\bv_1)]\\&\geq 2\delta,
    \end{align*}
where the last inequality holds due to $\wt(\bv_1) \leq k-\delta$.

    \textbf{Case 2:} $U \in \cC_3$.
    
    Suppose that the inverse identifying vector of $U$ is $\obu=(\overbrace{\obu_1}^{\mu_1}|\overbrace{\obu_3}^{\mu_3}|\overbrace{\obu_2}^{\mu_2})$ and the bilateral identifying vector of $V$ is $\tbv=(\overbrace{\bv_1}^{\mu_1} \mid \overbrace{\tbv_3}^{\mu_3} \mid \overbrace{\obv_2}^{\mu_2})$. Let $\cC_{\tcF_{\tbv}}$ be an $(\tcF_{\tbv}, \delta)_q$ GB-FD code, as defined in Theorem \ref{Our-T-GBPM} above.
    Assume that $\bV$ is the generator matrix of $V$, which is obtained by lifting a matrix $\bM \in \cC_{\tcF_{\tbv}}$. Write $\bV = (\bV_1 | \bV_2)$, where $\bV_2$ is the last $\mu_2$ columns of $\bV$. Hence, we can assume that the inverse identifying vector of $V$ is $\bar{\bv}=(\overbrace{\bar{\bv}_1}^{\mu_1}|\overbrace{\bar{\bv}_3}^{\mu_3}|\overbrace{\obv_2'}^{\mu_2})$ with $\wt(\bar{\bv}_2') =\rank(\bV_2)= \wt(\bar{\bv}_2) + \rank(\phi_{\tbv}(\bM))$. It is obvious that $\wt(\obu_2) = k $ and $\wt(\obu_1|\obu_3)=0$. By Lemma \ref{L-inv-L-S-H}, we have
    \begin{align*}
d_S(U,V) &\geq d_H(\bar{\bu},\bar{\bv})\\&= d_H(\bar{\bu}_1|\bar{\bu}_3,\bar{\bv}_1|\bar{\bv}_3)+d_H(\bar{\bu}_2, \bar{\bv}_2') \\&\geq |(k-\wt(\bar{\bu}_2)) - (k-\wt(\obv_2'))|
+
|\wt(\obv_2')-\wt(\obu_2)| \\&=
2 \cdot |\wt(\obv_2')-\wt(\obu_2)|\\
&=2\cdot|k-\wt(\bar{\bv}_2) - \rank(\phi_{\tbv}(\bM))| \\&\geq 2\delta,
    \end{align*}
    where the last inequality holds due to $\wt(\bv_1) = k- \wt(\bar{\bv}_2) \geq \delta$ and $\rank(\phi_{\tbv}(\bM))$ \\ $ \leq \wt(\bv_1)-\delta$.
    
    In conclusion, $\cC_1 \cup \cC_3 \cup \tcC_1$ is an $(n, |\cC_1|+|\cC_3|+|\tcC_1|, 
 2\delta, \{k\})_q$ CDC. This completes the proof.
\end{proof}
\begin{theorem}\label{SOur-T-GBPM}
    Take the same notation as in Corollary \ref{S-ParMixConstr}. Let $\mu_1=n_1+k-T_1^{\min}$, $\mu_2=k$, and $\mu_3=n-\mu_1-\mu_2$.
    Let $\tcS' \subseteq \BF_2^n$ be a set of bilateral identifying vectors  with  Hamming weight $k$, and minimum Hamming distance $2\delta$. Assume that all bilateral identifying vectors in $\tcS'$ have the same type $(\mu_1, \mu_3, \mu_2)$.
    Suppose that for any bilateral identifying vector $\tbv = (\bv_1 | \tbv_3 |\obv_2)$ in $\tcS'$, it satisfies $\delta \leq \wt(\bv_1) \leq k-\delta$.  If there exists an $(\tcF_{\tbv},\delta)_q$ GB-FD code $\cC_{\tcF_{\tbv}}$ such that $\rank(\phi_{\tbv}(\bM)) \leq \wt(\bv_1)-\delta$  for any $\tbv \in \tcS'$ and $\bM \in \cC_{\tcF_{\tbv}}$, then $\tcC_1' = \cup_{\tbv \in \tcS'} \sL(\cC_{\tcF_{\tbv}})$ is an $(n, 2\delta,\{k\})_q$ CDC. Moreover, $\cC_1 \cup \cC_3' \cup \tcC_1'$ is an $(n, |\cC_1|+|\cC_3'|+|\tcC_1'|, 2\delta, \{k\})_q$ CDC. 
\end{theorem}
\begin{proof}
    The proof is similar to that of Theorem \ref{Our-T-GBPM} and is omitted here.
\end{proof}

Next, we apply Theorems \ref{Our-T-GBPM} and \ref{SOur-T-GBPM} to drive some new lower bounds for CDCs.

\begin{theorem}\label{T-BPM-CDC}
    Let $n$, $n_1$, $n_2$, $k$, and $\delta$ be integers with $n=n_1+n_2,$ $k \geq 2\delta$,  and $ n_i \geq k$ for $i=1,2$. Let $T_1 \subseteq [\delta,k]$ with $l_{T_1} <2 \delta$. Let $T_2 \subseteq [k+\delta-T_1^{\min},k]$ with $l_{T_2} < 2\delta$. Let $n_3$ be a positive integer satisfying $n_1+k-T_1^{\min} \leq n-n_3+T_2^{\min}-k$.  Let $\mu_1=n_1+k-T_1^{\min}$, $\mu_2=n_3+k-T_2^{\min}$, and $\mu_3=n-\mu_1-\mu_2$.  Set $\omega_1 =\begin{cases}
    k-\delta & ~\text{if}~\mu_1 > \mu_2;\\
    \delta & \text{otherwise},
\end{cases}$ and $\omega_2=k-\omega_1$. Let $\theta_1 = \lfloor \frac{\mu_1-\omega_1}{\delta} \rfloor -1$ and $\theta_2= \min\{\lfloor \frac{\mu_2-\omega_2}{\delta} \rfloor, \lfloor \frac{\mu_2+\mu_3-\omega_2}{\delta} \rfloor -1\}$. Let $\Lambda_1^i =\max\{\omega_1, \mu_1-\omega_1-i\delta\}$, 
$\Gamma_1^i =\min\{\omega_1, \mu_1-\omega_1-i\delta\}$ for $0 \leq i \leq \theta_1$, and 
$\Lambda_3^j = \max\{\omega_2, \mu_2+\mu_3-\omega_2-j\delta\}$, $\Gamma_3^j = \min\{\omega_2, \mu_2+\mu_3-\omega_2-j\delta\}$ for $0 \leq j \leq \theta_2$. Set $d^{(i,j)}_1 =\begin{cases}
    \lceil \frac{\delta}{2} 
    \rceil & ~\text{if}~\Lambda_1^i \geq \Lambda_3^j;\\
    \lfloor \frac{\delta}{2} 
    \rfloor & \text{otherwise},
\end{cases}$ and $d^{(i,j)}_2=\delta-d^{(i,j)}_1$ for $0 \leq i \leq \theta_1$, $0 \leq j \leq \theta_2$. If there exists an $(n_1,2\delta,2\delta-l_{T_1},T_1)_q$ MDDC $\cX_1$ and an $(n_3,2\delta,2\delta-l_{T_2},T_2)_q$ MDDC $\cX_3$, then 
\begin{align*}
    A_q(n,2\delta,\{k\}) &\geq \sum_{t\in T_1} \eta_t(\cX_1) \cdot \Delta(k,n_2+t-k,\delta)_q \\& \quad + \sum_{s \in T_2}\eta_s(\cX_3) \cdot \Delta(k,n-n_3+s-k,\delta;k-\delta)_q \\& \quad + \Delta(\omega_1 , n_2-2k+T_1^{\min}+\omega_1, \delta;\omega_1-\delta)_q\\& \quad \cdot \sum_{i=0}^{\theta_1}\sum_{j=0}^{\theta_2}q^{\min\{\Lambda_1^i d^{(i,j)}_2,\Lambda_3^j d^{(i,j)}_1\}+\Lambda_1^i(\Gamma_1^i-\delta+1)+\Lambda_3^j(\Gamma_3^j -\delta+1)}.
\end{align*}
\end{theorem}

\begin{proof}
Set $\cC_1$ and $\cC_3$ to be the same as those in Theorem \ref{T-ParalMixConstr}. Then $|\cC_1 \cup \cC_3|=\sum_{t\in T_1} \eta_t(\cX_1)\Delta(k,n_2+t-k,\delta)_q + \sum_{s \in T_2}\eta_s(\cX_3) \Delta(k,n-n_3+s-k,\delta;k-\delta)_q$.

We choose a bilateral identifying vectors set $\tcS$ as follows:
  \begin{align*}  
   \tcS=\{&\tbv_{i,j}=(\underbrace{\overbrace{0\dots0}^{i\delta}\overbrace{1\dots1}^{\omega_1}0\dots0}_{\mu_1} \underbrace{\tilde{0}\dots\tilde{0}}_{\mu_3}
   \underbrace{\bar{0}\dots\bar{0} \overbrace{\bar{1}\dots\bar{1}}^{\omega_2} \overbrace{\bar{0}\dots\bar{0}}^{j\delta}}_{\mu_2}):\\& 0 \leq i \leq \theta_1, 0 \leq j \leq \theta_2\}.
  \end{align*}
It is obvious that the Hamming distance between any two distinct bilateral identifying vectors in $\tcS$ is at least $2\delta$. For any bilateral identifying vector $\tbv \in \cS$, we write $\tbv=(\overbrace{\bv_1}^{\mu_1}| \overbrace{\tbv_3}^{\mu_3}| \overbrace{\obv_2}^{\mu_2})$. Since $\wt(\bv_1) = \omega_1$ and $k \geq 2\delta$,  it is clear that $\delta \leq \wt(\bv_1) \leq k-\delta$.
For each bilateral identifying vector $\tbv_{i,j}$, where $0 \leq i \leq \theta_1$ and $0 \leq j \leq \theta_2$, its generalized bilateral echelon Ferrers form  is 
\begin{align*}
    \widetilde{\EF}(\tbv_{i,j}
)=
\left(
\begin{array}{ccc|c|ccc}
 \bO & \bI_{\omega_1} & \cF_1^i & \cF_2 & \cF_3^j &  \bO &\cF_4^j\\
 \bO & \bO & \bO  & \cF_5 & \bar{\cF}_6^j & \bar{\bI}_{\omega_2} & \bO
\end{array}
\right),
\end{align*}
and the generalized bilateral Ferrers diagram of $\tbv_{i,j}$ is
\begin{align*}
\tcF_{\tbv_{i,j}}=
\begin{array}{c|ccc}
        \cF_1^i & \cF_2 & \cF_3^j & \cF_4^j\\
        \hline
       \ &  \cF_5 & \bar{\cF}_6^j & \ 
\end{array},
\end{align*}
where $\cF_1^i$ is an $\omega_1 \times (\mu_1-\omega_1-i\delta)$ full Ferrers diagram, $\cF_2| \cF_3^j | \cF_4^j$ is an $\omega_1 \times (\mu_2+\mu_3-\omega_2)$ full Ferrers diagram, and $\cF_5|\bar{\cF}_6^j$ is an $\omega_2 \times (\mu_2+\mu_3-\omega_2 -j\delta)$ full Ferrers diagram.

By Proposition \ref{P-GBFD-Par}, for $0 \leq i \leq \theta_1$ and $0 \leq j \leq \theta_2$, we can construct an $(\tcF_{\tbv_{i,j}}, \rho_{i,j}, \delta)_q$ GB-FD code $\cC_{\tcF_{\tbv_{i,j}}}$ such that for any codeword $\bM \in \cC_{\tcF_{\tbv_{i,j}}}$, $\rank(\sigma(\bM)_{\omega_1, \mu_2+\mu_3-\omega_2}) \leq \omega_1-\delta$, where
\begin{align*}
\rho_{i,j}&=\min\{q^{\Lambda_1^i d^{(i,j)}_2}, q^{\Lambda_3^j d^{(i,j)}_1}\}\cdot q^{\Lambda_1^i(\Gamma_1^i-\delta+1)+\Lambda_3^j(\Gamma_3^j -\delta+1)} \\& \quad \cdot \Delta(\omega_1,\mu_2 +\mu_3-\omega_2, \delta;\omega_1-\delta)_q.
\end{align*}
Let $\tcC_1= \cup_{i=0}^{\theta_1} \cup_{j=0}^{\theta_2} \sL(\cC_{\tcF_{\tbv_{i,j}}})$ with cardinality $\sum_{i=0}^{\theta_1} \sum_{j=0}^{\theta_2} \rho_{i,j}$. By Theorem \ref{Our-T-GBPM}, $\cC_1 \cup \cC_3 \cup \tcC_1$ is an $(n,|\cC_1|+|\cC_3|+|\tcC_1|,2\delta, \{k\})_q$ CDC.
\end{proof}

 Below we give an example to illustrate Theorem \ref{T-BPM-CDC}.

\begin{example}
We adopt the notation in Theorem \ref{T-BPM-CDC}.
     Set $q=2$, $n=18$, $n_1=8$, $n_2=10$, $n_3=8$, $k=4$, $\delta=2$, $T_1=\{4,3\}$, and $T_2=\{4,3\}$. It is easy to check that $n_1+k-T_1^{\min} \leq n-n_3+T_2^{\min}-k$. If we take the $(8,4,3,\{4,3\})_2$ MDDC from Example \ref{E-Alg}, then in the proof of Theorem \ref{T-BPM-CDC}, we have
    \begin{align*}
        |\cC_1 \cup \cC_3|&=4801 \cdot 2^{30}+ 327 \cdot 2^{27}+4801 \cdot \Delta(4,10,2;2)_2 + 327 \cdot \Delta(4,9,2;2)_2\\
        &=5199101447408.
    \end{align*}
    Set $\mu_1=9$, $\mu_2=9$, $\mu_3=0$, $\omega_1=2$, $\omega_2=2$, and $\theta_1=\theta_2=2$. Let $\Lambda_1^i= 7-2i$, $\Lambda_3^j=7-2j$, $\Gamma_1^i=\Gamma_3^j =2$, and $d_1^{(i,j)}=d_2^{(i,j)}=1$ for $0 \leq i \leq 2$, $0 \leq j \leq 2$. 
Then in the proof of Theorem \ref{T-BPM-CDC}, we obtain that
\begin{align*}
    |\tcC_1|&=\Delta(2,7,2;0)_2 \cdot \sum_{i=0}^2\sum_{j=0}^2  2^{\min\{7-2i,7-2j\}+(7-2i)+(7-2j)}\\&=2^{21}+2\cdot 2^{17}+2^{15}+2\cdot2^{13}+2\cdot2^{11}+2^9\\&=2413056.
\end{align*}
As a result, we have
\begin{align*}
    A_2(18,4,\{4\}) &\geq |\cC_1 
    \cup \cC_3|+ |\tcC_1|=
5199101447408+ 2413056\\&
=5199103860464,
\end{align*}
which is better than the previously best-known lower bound $5199101447408$ in \cite{ParllelMixed}.
\end{example}

\begin{theorem}\label{T-SPar-CDC}
Let $n, n_1, n_2, k,$ and $\delta$ be integers with $n=n_1+n_2$, $n_1 \geq k$, $n_2\geq k$, 
 and $k \geq 2\delta \geq 4$. Let $T_1 \subseteq [\delta,k]$ with $l_{T_1} < 2\delta$. Assume that $n_1+k-T_1^{\min} \leq n-k$. Let $\theta_1=\lfloor\frac{n_1-T_1^{\min}}{\delta} \rfloor$ and $\theta_2 = \min\{\lfloor \frac{k}{\delta} \rfloor-1, \lfloor \frac{n_2-k+T_1^{\min}}{\delta} \rfloor-2\}$.  Let $\Lambda_1^i= \max\{k-\delta, n_1-T_1^{\min} +(1-i)\delta \}$, $\Gamma^i=\min\{k-\delta,n_1-T_1^{\min}+(1-i)\delta\}$ for $0 \leq i \leq \theta_1$, and  $\Lambda_3^j=n_2-k+T_1^{\min}-(j+1)\delta$ for $0 \leq j \leq \theta_2$. Set $d^{(i,j)}_1 =\begin{cases}
    \lceil \frac{\delta}{2} 
    \rceil & ~\text{if}~\Lambda_1^i \geq \Lambda_3^j;\\
    \lfloor \frac{\delta}{2} 
    \rfloor & \text{otherwise},
\end{cases}$ and $d^{(i,j)}_2=\delta-d^{(i,j)}_1$ for $0 \leq i \leq \theta_1, 0 \leq j \leq \theta_2$. If there exists an $(n_1, 2\delta, 2\delta-l_{T_1}, T_1)_q$ MDDC $\cX_1$, then
 \begin{align*}
     A_q(n,2\delta,\{k\}) &\geq \sum_{t\in T_1}\eta_t(\cX_1) \cdot \Delta(k,n_2+t-k, \delta)_q+\Delta(k,n-k, \delta; k-\delta)_q\\& \quad + \Delta(k-\delta, n_2-k+T_1^{\min}-\delta, \delta; k-2\delta)_q \\ &\quad \cdot \sum_{i=0}^{\theta_1} \sum_{j=0}^{\theta_2}q^{\min\{\Lambda_1^i d^{(i,j)}_2,\Lambda_3^j d^{(i,j)}_1\}+ \Lambda_1^i(\Gamma^i-\delta+1)+\Lambda_3^j}.
 \end{align*}
\end{theorem}
\begin{proof}
Let $\cC_1$ and $\cC_3'$ be constructed from Corollary \ref{S-ParMixConstr}. Then $|\cC_1 \cup \cC_3'|= \sum_{t\in T_1}\eta_t(\cX_1)\Delta(k,n_2+t-k, \delta)_q+\Delta(k,n-k, \delta; k-\delta)_q$.

    We choose a set of bilateral identifying vectors $\tcS'$ as follows:
  \begin{align*}  
   \tcS'=\{&\tbv_{i,j}=(\underbrace{\overbrace{0\dots0}^{i\delta}\overbrace{1\dots1}^{k-\delta}0\dots0}_{n_1+k-T_1^{\min}} \underbrace{\tilde{0}\dots\tilde{0}}_{\mu}
   \underbrace{\bar{0}\dots\bar{0} \overbrace{\bar{1}\dots\bar{1}}^{\delta} \overbrace{\bar{0}\dots\bar{0}}^{j\delta}}_{k}):\\& 0 \leq i \leq \theta_1, 0 \leq j \leq \theta_2\},
  \end{align*}
  where $\mu=n_2-2k+T_1^{\min}$.
  Then for $0 \leq i \leq \theta_1$ and $0 \leq j \leq \theta_2$, the generalized bilateral echelon Ferrers form $\widetilde{\EF}(\tbv_{i,j})$ of the bilateral identifying vector $\tbv_{i,j}$ is
\begin{align*}
    \widetilde{\EF}(\tbv_{i,j}
)=
\left(
\begin{array}{ccc|c|ccc}
 \bO & \bI_{k-\delta} & \cF_1^i & \cF_2 & \cF_3^j &  \bO &\cF_4^j\\
 \bO & \bO & \bO  & \cF_5 & \bar{\cF}_6^j & \bar{\bI}_{\delta} & \bO
\end{array}
\right),
\end{align*}
and the generalized bilateral Ferrers diagram of $\tbv_{i,j}$ is
\begin{align*}
\tcF_{\tbv_{i,j}}=
\begin{array}{c|ccc}
        \cF_1^i & \cF_2 & \cF_3^j & \cF_4^j\\
        \hline
       \ &  \cF_5 & \bar{\cF}_6^j & \ 
\end{array},
\end{align*}
where $\cF_1^i$ is a $(k-\delta) \times (n_1-T_1^{\min}-(i-1)\delta)$ full Ferrers diagram, $\cF_2|\cF_3^j|\cF_4^j$ is a $(k-\delta) \times (n_2-k+T_1^{\min}-\delta)$ full Ferrers diagram, and $\cF_5|\bar{\cF}_6^j$ is a $\delta \times (n_2-k+T_1^{\min} -(1+j)\delta)$ full Ferrers diagram.

By Proposition \ref{P-GBFD-Par}, for $0 \leq i \leq \theta_1$ and $0 \leq j \leq \theta_2$, we can construct an $(\tcF_{\tbv_{i,j}}, \rho_{i,j}, \delta)_q$ GB-FD code $\cC_{\tcF_{\tbv_{i,j}}}$ such that for any codeword $\bM \in \cC_{\tcF_{\tbv_{i,j}}}$, $\rank(\sigma(\bM)_{k-\delta, n_2-k+T_1^{\min}-\delta}) \leq k-2\delta$, where
\begin{align*}
\rho_{i,j}&=\min\{q^{\Lambda_1^i d^{(i,j)}_2}, q^{\Lambda_3^j d^{(i,j)}_1}\}\cdot q^{\Lambda_1^i(\Gamma^i-\delta+1)+\Lambda_3^j} \\&\quad \cdot \Delta(k-\delta, n_2-k+T_1^{\min}-\delta, \delta;k-2\delta)_q.
\end{align*}
Let $\tcC_1'= \cup_{i=0}^{\theta_1} \cup_{j=0}^{\theta_2} \sL(\cC_{\tcF_{\tbv_{i,j}}})$ with cardinality $\sum_{i=0}^{\theta_1} \sum_{j=0}^{\theta_2} \rho_{i,j}$. By Theorem \ref{SOur-T-GBPM}, $\cC_1 \cup \cC_3' \cup \tcC_1'$ is an $(n,|\cC_1| + |\cC_3'|+ |\tcC_1'|, 2\delta, \{k\})_q$ CDC.
\end{proof}

 Below we give an example to illustrate Theorem \ref{T-SPar-CDC}.

\begin{example}
We adopt the notation in Theorem \ref{T-SPar-CDC}.
Set $q=2$, $n=15$, $n_1=8$, $n_2=7$, $k=4$, $\delta=2$, and $T_1=\{4,3\}$. It is easy to verify that $n_1+k-T_1^{\min} \leq n-k$. Set  $\theta_1=2$ and $\theta_2=1$.
Set   $\Lambda_1^i=7-2i$, $\Gamma^i=2$ for $0 \leq i \leq 2$, and $\Lambda_3^j=4-2j$ for $0 \leq j \leq 1$. Set $d^{(i,j)}_1=d^{(i,j)}_2=1$ for $0 \leq i \leq 2$ and $0 \leq j \leq 1$. If we take the $(8,4,3,\{4,3\})_2$ MDDC from Example \ref{E-Alg}, then $ |\cC_1 \cup \cC_3'| = 10154219486$ and  $ |\tcC_1'| =44672$ in the proof of Theorem \ref{T-SPar-CDC}. As a result, we have
    \begin{align*}
        A_2(15,4,\{4\}) \geq |\cC_1 \cup \cC_3'|+ |\tcC_1'|=10154219486+44672=10154264158,
    \end{align*}
     which is better than the previously best-known lower bound $10154219486$ in \cite{ParllelMixed}.
\end{example}
%%%%%%%%%%%%

\section{New lower bounds for constant dimension codes}\label{Sec-Result}

In this section, we compare our lower
bounds for CDCs derived from Theorems \ref{T-BPM-CDC} and \ref{T-SPar-CDC} with the best-known lower bounds for CDCs in \cite{ParllelMixed}.

\begin{corollary}\label{2-12}
Let $h$ be a positive integer. Then
\begin{align*}
    A_2(12+h,4,\{4\}) &\geq 4801\cdot 2^{12+3h} + 327 \cdot 2^{9+3h} \\& \quad +(1+(2^{(8+h)}-1)\cdot 35)\\& \quad + \sum_{i=0}^2 \sum_{j=0}^{\min\{1,\lfloor \frac{h-1}{2} \rfloor\}}2^{\min\{7-2i,h+1-2j\}+8+h-2i-2j}.
\end{align*} 
\end{corollary}
\begin{proof}
    Let $q=2$, $n=12+h$, $n_1=8$, $n_2=4+h$, $k=4$, $\delta=2$, and $T_1=\{4,3\}$. Obviously, $n_1+k-T_1^{\min} \leq n-k$. Set $\theta_1=2$ and $\theta_2=\min\{1, \lfloor \frac{h-1}{2} \rfloor\}$. Let $\Lambda_1^i =7-2i$, $\Lambda_3^j= 1+h-2j$, $\Gamma^i=2$, and $d_1^{(i,j)}=d_2^{(i,j)}=1$ for $0 \leq i \leq \theta_1$ and $0 \leq j \leq  \theta_2$. If we take the $(8,4,3,\{4,3\})_2$ MDDC from Example \ref{E-Alg}, then by Theorem \ref{T-SPar-CDC}, the conclusion follows.
\end{proof}

\begin{corollary}\label{2-18}
    Let $h$ be an integer with $h \geq 0$. Then
    \begin{align*}
        A_2(18+h, 4, \{4\}) &\geq 4801 \cdot 2^{30+3h}+327 \cdot 2^{27+3h} \\&\quad +4801 \cdot (1+35(2^{10+h}-1) ) +327\cdot(1+35(2^{9+h}-1))\\&\quad +  \sum_{i=0}^2 \sum_{j=0}^{\min\{3,\lfloor \frac{5+h}{2} \rfloor \}}  2^{\min\{7-2i, 7-2j+h\}+14+h-2i-2j}.
    \end{align*}
\end{corollary}
\begin{proof}
    Let $q=2$, $n=18+h$, $n_1=8$, $n_2=10+h$, $n_3=8$, $k=4$, $\delta=2$, $T_1=\{4,3\}$, and $T_2=\{4,3\}$. It is easy to verify that  $n_1+k-T_1^{\min} \leq n-n_3+T_2^{\min}-k$.  Set $\mu_1=9$, $\mu_2=9$, $\mu_3=h$, $\omega_1=\omega_2=2$, $\theta_1=2$, and $\theta_2=\min\{3,\lfloor \frac{5+h}{2} \rfloor\}$. Let $\Lambda_1^i=7-2i$, $\Lambda_3^j=7-2j+h$, $\Gamma_1^i=\Gamma_3^j=2$, and $d_1^{(i,j)}=d_2^{(i,j)}=1$ for $0 \leq i \leq \theta_1$, $0 \leq j \leq  \theta_2$. If we take the $(8,4,3,\{4,3\})_2$ MDDC from Example \ref{E-Alg}, then by Theorem \ref{T-BPM-CDC}, the conclusion follows.
\end{proof}

The following theorem gives a lower bound for the size of $(n,2\delta,\delta+1)_q$ MDDCs, which is useful for our constructions of CDCs.
\begin{theorem}[\cite{MixDD}]\label{Constr-MDDC}
    Let $n$, $k$, and $\delta$ be integers with $\delta \geq 2$ and $n > k \geq 2\delta-1$. Let $\cC_0$ be an $(n, |\cC_0|, 2\delta, \{k\})_q$ CDC, and \[ N= \left\lfloor 
     \frac{[\begin{smallmatrix}
         n\\ k-\delta+1
     \end{smallmatrix}]_q - |\cC_0|[\begin{smallmatrix}
         k\\\delta-1
         \end{smallmatrix}]_q}{\sum_{i=0}^{\delta-1}q^{i^2}[\begin{smallmatrix}
             k-\delta+1\\i
         \end{smallmatrix}]_q [\begin{smallmatrix}
             n-(k-\delta+1)\\ i
         \end{smallmatrix}]_q} \right\rfloor.\]
         Then there exists an $(n, 2\delta, \delta+1)_q$ MDDC $\cC$ with $\eta_k(\cC) = |\cC_0|$ and $\eta_{k-\delta+1}(\cC)= \max\{N,0\}$.
\end{theorem}
In the following corollaries, 
we denote \[N_q(x,y,z)= \left\lfloor 
    \frac{[\begin{smallmatrix}
         x\\ z-y+1
     \end{smallmatrix}]_q - S_q(x,2y,z)[\begin{smallmatrix}
         z\\ y-1
         \end{smallmatrix}]_q}{\sum_{i=0}^{y-1}q^{i^2}[\begin{smallmatrix}
             z-y+1\\ i
         \end{smallmatrix}]_q [\begin{smallmatrix}
             x-(z-y+1)\\ i
         \end{smallmatrix}]_q} \right\rfloor,\]
where $S_q(x,2y,z)$ represents the best-known size of the $(x,2y, \{z\})_q$ CDC.

With the help of Theorem \ref{Constr-MDDC}, we can obtain general results of Corollaries \ref{2-12} and \ref{2-18}.

\begin{corollary}\label{q-18}
    Let $h$ be an integer with $h \geq 0$. Then
    \begin{align*}
        A_q(18+h,4,\{4\}) 
  &\geq A_q(8,4,\{4\}) \cdot q^{30+3h}+N_q(8,2,4) \cdot q^{27+3h} \\& \quad + A_q(8,4,\{4\}) \cdot (1+(q^{10+h}-1)(q^2+1) (q^2+q+1))\\& \quad +N_q(8,2,4)\cdot(1+(q^{9+h}-1)(q^2+1)(q^2+q+1))\\& \quad+  \sum_{i=0}^2 \sum_{j=0}^{ \min\{3,\lfloor\frac{5+h}{2} \rfloor \}} q^{\min\{7-2i, 7-2j+h\}+14+h-2i-2j}.
    \end{align*}
\end{corollary}
\begin{proof}
The proof of the Corollary \ref{q-18} is similar to the proof of Corollary \ref{2-18}. Construct the $(8,4,3,\{4,3\})_q$ MDDC from Theorem \ref{Constr-MDDC}, and then the conclusion follows by Theorem \ref{T-BPM-CDC}.
\end{proof}

\begin{corollary}\label{q-12}
Let $\delta \geq 2$ and $h \geq \delta-1$ be two integers. \\
Set $d^{(i,j)}_1 =\begin{cases}
    \lceil \frac{\delta}{2} 
    \rceil & ~\text{if}~j-i \geq \frac{h+2}{\delta}-4;\\
    \lfloor \frac{\delta}{2} 
    \rfloor &~\text{otherwise},
\end{cases}$ and  $d^{(i,j)}_2=\delta-d^{(i,j)}_1$ for $0 \leq i \leq 2$ and $0 \leq j \leq  \min\{1,\lfloor \frac{h+1}{\delta} 
\rfloor -1\}$. If $N_q(4\delta,\delta,2\delta) > 0$, then 
\begin{align*} A_q(6\delta+h,2\delta,\{2\delta\}) &\geq A_q(4\delta,2\delta,\{2\delta\})\cdot q^{(2\delta+h)(\delta+1)}\\&\quad+N_q(4\delta,\delta,2\delta)\cdot q^{\max\{\delta+h+1,2\delta\}(\min\{\delta+h+1,2\delta\}-\delta+1)} \\&\quad +(1+\left[\begin{smallmatrix}
    2\delta\\ \delta
\end{smallmatrix}\right]_q \cdot (q^{4\delta+h}-1)) \\& \quad +\sum_{i=0}^2 \sum_{j=0}^{\min\{1,\lfloor \frac{h+1}{\delta} \rfloor-1\}} [q^{\min\{((4-i)\delta-1)d_2^{(i,j)},(h+1-j\delta)d_1^{(i,j)}\}} \\&\quad \cdot q^{(4-i)\delta+h-j\delta}].
\end{align*}
\end{corollary}
\begin{proof}
    Set $n=6\delta+h$, $n_1=4\delta$, $n_2=2\delta+h$, $k=2\delta$, $d=2\delta$, and $T_1=\{2\delta,\delta+1\}$. Since $h \geq \delta-1$, it is obvious that $n_1+k-T_1^{\min} \leq n-k$. Set $\theta_1=2$ and $\theta_2=\min\{1,\lfloor \frac{h+1}{\delta} \rfloor -1\}$. Set $\Lambda_1^i=(4-i)\delta-1$, $\Gamma^i=\delta$,  for $0 \leq i \leq \theta_1$, and $\Lambda_3^j=h+1-j\delta$ for $0 \leq j \leq \theta_2$. If we construct the $(4\delta,2\delta,\delta+1, T_1)_q$ MDDC from Theorem \ref{Constr-MDDC}, then the conclusion follows by Theorem 
\ref{T-SPar-CDC}.
\end{proof} 

\begin{example}
    To illustrate the difference between our constructions and the parallel mixed dimension codes in \cite{ParllelMixed}, take $h=2$ and $\delta=2$ in Corollary \ref{q-12}. By the parallel mixed dimension construction, we have
    \begin{align*}
        A_q(14,4,\{4\})\geq A_q(8,4,\{4\}) \cdot q^{18}+N_q(8,2,4) \cdot q^{15}+(1+\left[\begin{smallmatrix}
    4\\ 2
\end{smallmatrix}\right]_q \cdot (q^{10}-1)).
    \end{align*}
    While from our construction, we have 
    \begin{align*}
        A_q(14,4,\{4\})&\geq A_q(8,4,\{4\}) \cdot q^{18}+N_q(8,2,4) \cdot q^{15}+(1+\left[\begin{smallmatrix}
    4\\ 2
\end{smallmatrix}\right]_q \cdot (q^{10}-1))\\&
+q^{13}+q^{11}+q^9.
    \end{align*}
Here, $\left[\begin{smallmatrix}
    4\\ 2
\end{smallmatrix}\right]_q=(q^2+1)(q^2+q+1)$. 
    Therefore, 
    \begin{align*}
       A_q(14,4,\{4\}) &\geq  A_q(8,4,\{4\}) \cdot q^{18}+N_q(8,2,4) \cdot q^{15}+q^{14} + 2q^{13} \\&+ 2q^{12} + 2q^{11} + q^{10} + q^9 - q^4 - q^3 - 2q^2 - q.
    \end{align*}
\end{example}

Corollaries \ref{2-12}, \ref{2-18}, \ref{q-18}, and \ref{q-12}
provide many new CDCs with larger sizes
than those of the previously best-known codes in \cite{ParllelMixed}. We list some of our new lower bounds in Table \ref{NewBound}. Overall at least $49$ new lower bounds are provided. Table \ref{NewBound1} shows some examples of these new lower bounds.

\begin{table}[!htb]
\caption{New lower bounds of CDCs}

\label{NewBound}

\begin{tabular*}{\hsize}{@{}@{\extracolsep{\fill}}ll@{}}
 \toprule
New lower bounds $(q=2,3,4,5,7,8,9)$ & Corollary \\
\midrule
$A_2(12+h,4,\{4\})~(1 \leq h \leq 5)$ & Corollary \ref{2-12} \\ 

  $A_2(18+h,4,\{4\})~(0 \leq h \leq 1)$ & Corollary \ref{2-18} \\ 
  
  $ A_q(n,4,\{4\})~(q\geq3, 13 \leq n \leq 17)$ & Corollary \ref{q-12} \\ 
  
 $ A_q(n,4,\{4\})~(q\geq3, n=18,19)$ & Corollary \ref{q-18} \\ \bottomrule
\end{tabular*}
\end{table}

\begin{table}[!htp]
  \caption{Comparison of cardinalities of our codes with the codes in \cite{ParllelMixed}  when $q=2, 3$}
  \label{NewBound1}
  \begin{tabular*}{\hsize}{@{}@{\extracolsep{\fill}}lll@{}}
  \toprule
   $A_q(n,d,k)$ & New & Old \\ \midrule
  $A_2(13,4,\{4\})$ & 158679134 & 158676446\\ 
  
  $A_2(14,4,\{4\})$& 
  1269315038 &1269304286\\

  $A_2(15,4,\{4\})$& 10154264158&10154219486 \\
  
  $A_2(16,4,\{4\})$& 81233502686&81233326046 \\
  
  $A_2(17,4,\{4\})$& 649866384350&649865748446 \\ 
  
  $A_2(18,4,\{4\})$& 5199103860464&5199101447408\\
  
  $A_2(19,4,\{4\})$& 41591750345072&41591745233648\\
  
  $A_3(13,4,\{4\})$ & 7793875720905& 7793875521888\\ 
  
  $A_3(14,4,\{4\})$ & 210434579474523&
  210434577683370 \\ 
  
  $A_3(15,4,\{4\})$ & 
  5681733429422760&
  5681733413221464 \\ 
  
  $A_3(16,4,\{4\})$ & 153406801749741632
  &
  153406801604284256 \\ 
  
  $A_3(17,4,\{4\})$ &4141983642879423488
  &4141983641657581568 \\ 
  
  $A_3(18,4,\{4\})$ & 111833562501139316736&111833562490271858688 \\ 
  
  $A_3(19,4,\{4\})$ &3019506087163758772224
&3019506087130763231232 \\ \bottomrule
  \end{tabular*}
\end{table}

\section{Conclusion}\label{Sec-Conclusion}

In this paper, we apply the generalized bilateral multilevel construction to the parallel mixed dimension construction. Our method leads to new lower bounds of CDCs for all feasible parameters in the initial parallel mixed dimension construction \cite{ParllelMixed}. MDDCs and GB-FD codes play important roles in the parallel mixed dimension construction and the generalized bilateral multilevel construction, respectively. For further research, constructing larger MDDCs and GB-FD codes can lead to better lower bounds for CDCs. Our methods focus on constructing CDCs for $k\geq d$. In the future, we aim to extend this to the case where $k < d$.

\end{document}